\documentclass[a4paper,onecolumn,11pt]{article}

\usepackage{graphicx}          
\usepackage{amsmath}           
\usepackage{amssymb}
\usepackage{algorithm}
\usepackage{algpseudocode}
\usepackage{hyperref}
\usepackage{color}
\usepackage{cite}
\usepackage{enumerate}
\usepackage{epstopdf}
\usepackage{subcaption}
\usepackage{appendix}

\usepackage[latin1]{inputenc} 
\usepackage{fancyhdr}
\usepackage{amsthm}
\newtheorem{theorem}{Theorem}[section]

 \theoremstyle{definition}
\newtheorem{definition}{Definition}[section]
\theoremstyle{remark} \newtheorem{remark}{Remark}

\newcommand{\mbf}[1]{\mathbf{#1}}
\newcommand{\mbs}[1]{\boldsymbol{#1}}

\begin{document}


\title{\textbf{Leontief Meets Shannon -- Measuring the Complexity of the
  Economic System}}
\author{Dave Zachariah and Paul Cockshott\footnote{E-mail:
    \texttt{dave.zachariah@it.uu.se} and \texttt{william.cockshott@glasgow.ac.uk}}
  \thanks{Thanks to Nathaniel Lane at the Institute for   International Economic Studies, Stockholm University, for kindly
  providing data on the South Korean economy.}}
\date{} 
\maketitle

\begin{abstract}
We develop a complexity measure for large-scale economic
  systems based on Shannon's concept of entropy. By adopting Leontief's
  perspective of the production process as a circular flow, we
  formulate the process as a Markov chain. Then we derive a
  measure of economic complexity as the average number of bits required
  to encode the flow of goods and services in the production process. We
  illustrate this measure using data from seven national economies, spanning
  several decades.
\end{abstract}

\section{Introduction}

 It appears self evident that an industrial economy is 
 more complex than a pre-industrial one, but how are we
 to measure this complexity?
 
 Is the complexity a matter of the number of people involved?
 That probably has something to do with it since the population
 density in Europe today is well above that in Europe 500 years
 ago. But that scarcely seems enough. Clearly along with the
 increase in the number of people has gone an increase in the
 number of distinct products, and associated with that an
 increase in thenumber of trades or specialisations that people occupy. 
 The classical economists viewed this social division of labour as a
 critical step in raising the productive capacity of an economy \cite{smith1974}.
 
 An initial approach might be to quantify complexity 
 simply as the number of trades followed. But that does not seem
 enough either. Suppose we have two economies, each with
 one hundred distinct trades. In the first economy, 90\% of the population
 are still farmers, with the special trades being distributed
 among the the 10\% townsfolk. In the second economy, the great
 majority are urban with the population more evenly distributed
 between trades. It would be reasonable to say that
 the 90\% rural country was economically simpler and less complex.
 
 So intuitively, complexity ought to depend both on the number of trades
 and the distribution of the population into them. One way
 of approaching it would be to say that in a complex economy
 more information has to be provided to describe the average
 life outcome of an individual: will they be a peasant, a smith,
 a taylor, a wheelwright etc. This notion is possible to
 operationalize, if we view the economy as a process
 which randomly assigns individuals places in the social division of
 labour. The occupations each individual ends up in is then a random
 variable $s$ drawn from a set of distinct trades.   
 Then we know from Shannon
 \cite{shannon1948} that the average amount of information, (in bits) required
 to describe where the individuals end up is given by the entropy, denoted $H(s)$.

The entropy has several nice properties as a measure
\cite{Palan2010_specializationindices}. It has a lower bound 0 and an
upper bound for a given a number of trades or economic sectors. It is also
invariant to the order in which one lists economic sectors, assuming
that the employment in each sector is given. However, $H(s)$ does not take into account that the economic
system constitutes a process with interconnected sectors. Twelve
complexity measures based on sectoral connectedness were surveyed in
\cite{LopesEtAl2008_survey}. Some of them are bounded but provide no
interpretable notion of complexity nor do they immediately relate to
the economic production process. Alternative notions of complexity
formulated from an algorithmic point of view do, however, provide an
interpretable meaning. One notion is to
measure the complexity of the interconnections of the economy in terms of the shortest algorithm
that could generate the economic structure
\cite{chaitin:1999}. Another is to measure the number of steps required for an
economy to solve its resource allocation problem
\cite{cockshott97}. These notions are dependent on how closely tied up
the economic sectors are. The more they are interlinked,
the greater the potential set of interactions, and the greater the
algorithmic complexity. If there are $n$ distinct sectors, then the
runtime complexity of solving for equilibrium prices in a fully connected
economy is on the order of $n^3$. The fact that market-based
economies with very large $n$ manage to coordinate production 
via a price mechanism, suggests that the complexity
 is significantly lower than this. One reason is that the
 interconnections between sectors in real economies are typically
 sparse, being more like Erd\"os-R\'enyi graphs 
 than fully connected ones \cite{reifferscheidt2014average}.

Leontief was one of the pioneers in developing input-output models of the entire
economic production system, which are essential in the construction of
national accounts data
\cite{Leontief1928_economy,Leontief1941_structure,Stone1961_input}. Starting from Leontief's
fundamental insights, we develop a complexity measure of an economic
system using ideas originating from Markov
\cite{Markov1906_rasprostranenie} and Shannon \cite{shannon1948}. The
measure quantifies the average number of bits required to describe the
flow of goods and services, from one sector to the next, when viewing
production as a continual process. This provides both a practical
interpretation, rooted in the engineering sciences, as well as a
notion related to the production process, unlike the measures surveyed
in \cite{LopesEtAl2008_survey}. The measure is applied to data from seven national economies, spanning several decades.

\section{Economic complexity}

From Leontief's point of view, the economic production system has
multiple inputs and multiple outputs of distinct goods and services
\cite{Leontief1986_io,Pasinetti1977_lectures,TenRaa2006_economics}. Under capitalist institutions,
it takes the form of production of commodities by means of
commodities \cite{sraffa}. To begin the analysis of such a system,
we partition the economy into sectors, each producing a distinct type of
good or service. A fraction
of output from sector $i$ is required for production in sector
$j$. A certain fraction is also used up internally. Sectors with outputs that
enter directly and indirectly in the production of all other sectors
are denoted `basic' \cite[ch.~2]{sraffa}. The basic sectors are labeled $i=1, \dots, n$. Together they form
an interconnected reproducing
economic system, which we denote as $\mathcal{S}$ and corresponds to
an irreducible and aperiodic graph. 

For sake of illustration, suppose labour is required to
initialize the production process and let the direct labour requirement
in each sector be denoted as $\ell_i$. Consider a unit of the total social
labour to be allocated to an initial sector
$$s_0 \in \{ 1, \dots, n\}$$
randomly. Then $\pi_i = \ell_i / \sum_j \ell_j$ is the probability
that a unit of labour is allocated to produce in sector $i$. That is,
$\Pr\{ s_0 = i \}$ equals $\pi_i$. Using this formalism, we can study the economic requirements of
production as a random process. The entropy of the random initial sector $s_0$ is then given by
\begin{equation}
  H(s_0) = -\sum^n_{i=1} \pi_i \: \log_2 \pi_i \: \geq 0,
\label{eq:0-entropy}
\end{equation}
which gives the average number of bits required to encode the initial
sector \cite{shannon1948}.\footnote{By convention, we have $0
  \log_2(0) = 0$, which is justified by continuity
  \cite{Cover&Thomas2012_elements}. This addresses the concern raised
  in \cite{Palan2010_specializationindices} for sectors where $\pi_i=0$.} To illustrate the economic meaning of this quantity in the case of labour, consider a
pre-industrial economy. Here the great bulk of the working population
are peasants and only a small proportion are employed in
non-agricultural sectors. In consequence, the entropy $H(s_0)$ in the pre-industrial
economy is low. Industrialisation takes people from the countryside,
and randomly casts them into a plethora of urban trades. This
transition process will then increase the entropy $H(s_0)$. The
maximum entropy is attained if all sectors require an equal amount of
labour. Then $H(s_0) = \log_2(n)$.

Production in sector $s_0$ yields outputs that are requirements in
a subsequent sector $s_1$. Analogous to the reasoning above, we
consider this as a random transition
$s_0 \rightarrow s_1$, which enables a stochastic formulation of Leontief's input-output model. Let $$f_{ij} \geq 0, \qquad \forall i, j \in \{1, \dots, n \}$$
denote the cross-sectoral requirements of output $i$ in sector $j$. 
\begin{definition}
An
$n \times n$ matrix $\mbf{P} = \{ p_{ij} \}$, where
$$p_{ij} \triangleq \frac{f_{ij}}{\sum^n_{j=1} f_{ij}}$$
is the fraction of sector $i$:s cross-sectoral output required for production in
sector $j$.
\end{definition}
 The production
requirements, in the form of forwarding cross-sectoral outputs from one
sector to the next, can now be modeled as a Markov chain: Consider a
unit of cross-sectoral output from $s_0$. It becomes an input in
sector $j$ with a probability $p_{ij}$. Thus we consider the
random transition $s_0 \rightarrow s_1$ to occur with a probability  $\Pr\{ s_1
= j  \: | \: s_0 = i \} = p_{ij}$. See Figure~\ref{fig:markov} for an
illustration.
\begin{figure}
  \begin{center}
    \includegraphics[width=0.28\columnwidth]{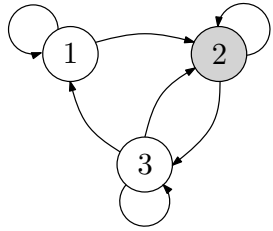}
  \end{center}
  \caption{Example of economic system $\mathcal{S}$ with $n=3$ sectors
  producing distinct outputs. The initial sector is $s_0 = 2$
  (highlighted). Each link is associated with a probability $p_{ij}$ that unit of output $i$ will be used up in sector $j$. This
  formalizes the production process as a transition from one sector to
the next.}
  \label{fig:markov}
\end{figure}
The entropy of the subsequent sector $s_1$ given the
prior sector is
\begin{equation}
  H(s_1 | s_0 = i ) = -\sum^n_{j=1} p_{ij}  \:  \log_2 p_{ij}.
\end{equation}
To illustrate the economic meaning of this quantity, consider an
industrial economy prior to electrification and suppose the output of sector $i$
is electricity generated by coal power plants. Initially this
output enters into the production of only a few sectors. In
consequence, the entropy $H(s_1 | s_0 = i)$ prior to electrification is low. With the introduction of
transmission lines and power grids, electricity becomes widely used
and $H(s_1 |
s_0 = i )$ rises. The number of bits required to encode $s_1$ given a
random initial state $s_0$ is obtained by simple averaging,
\begin{equation}
H(s_1 | s_0) = \sum^n_{i=1} \pi_i H(s_1 | s_0 = i ) .
\label{eq:1-entropy}
\end{equation}
Thus as the economy diversifies and more sectors require a complex mix
of inputs, the conditional entropy $H(s_1 | s_0)$ will increase. Note
that since the measure is logarithmic, an increment of one full bit corresponds to a doubling of complexity.

We can now extend above approach to trace a sequence of unit outputs
through subsequent sectors in  the production process. That is, the sequence of
$(k+1)$ random variables,
\begin{equation}
s_0 \: \rightarrow \: s_1 \: \rightarrow \: \cdots \: \rightarrow \:
s_k,
\label{eq:sequence}
\end{equation}
that represent the $k$th order production requirements of the economy
in a forward direction. The conditional entropy of the current sector
$s_k$ in the sequence is denoted
\begin{equation}
H(s_k | s_{k-1}, \dots, s_0 )
\label{eq:k-entropy}
\end{equation}
and can be derived by extension of the first-order conditional entropy
\eqref{eq:1-entropy}. As Leontief realized, the production process is a `circular flow' of
requirements \cite{Leontief1928_economy}. The sequence
\eqref{eq:sequence} is a therefore a representation of the entire 
process when $k$ tends to infinity. Given the
the properties of the sectors in $\mathcal{S}$, the distribution of
$s_k$ will then tend to a unique stationary distribution
\cite[ch.~4]{Cover&Thomas2012_elements} and under
stationarity, the conditional entropy of the sequence is a nonincreasing
quantity, i.e.,
\begin{equation*}
\begin{split}
H(s_k | s_{k-1}, \dots, s_0 ) &\leq H(s_k | s_{k-1}, \dots, s_1 ) \\
&= H(s_{k-1} | s_{k-2}, \dots, s_0 ).
\end{split}
\end{equation*}
In other words, while the production process is conceptualized as an unceasing circular flow, it is
possible to define a finite limit of its conditional entropy.
\begin{definition}[Economic complexity]
The entropy of the economic
  system $\mathcal{S}$ is defined as
$$H(\mathcal{S} ) \; \triangleq \; \lim_{k \rightarrow \infty} \; H(s_k | s_{k-1}, \dots, s_0 )
,$$
and represents the number of bits required to encode $s_k$ in the limit.
\end{definition}
The entropy $H(\mathcal{S} )$ is the average code length required to encode the
destination of an output in the production process. Thus it is a
natural measure of complexity of $\mathcal{S}$ in logarithmic scale. For any given
entropy, the economic system must effectively discriminate between
$2^{H(\mathcal{S} )}$ output destinations at each step of the process. The entropy of two systems,
$\mathcal{S}$ and $\mathcal{S}'$, enable a comparison across
time and space. Suppose $\mathcal{S}$ and $\mathcal{S}'$ both produce
basic four-wheeled automobiles, and that its production in $\mathcal{S}$
requires a complex mix of inputs, each of which in turn require
another a complex mix to produce. If $\mathcal{S}'$ produces
automobiles in a simpler manner, then $H(\mathcal{S})$ will be greater than 
$H(\mathcal{S}')$. Since the relative difference in the effective number of output
destinations is $2^{H(\mathcal{S} ) - H(\mathcal{S'} )}$, each
additional bit to $\mathcal{S}$ represents a doubling of the
complexity of the economy $\mathcal{S}$ relative to $\mathcal{S}'$.
\begin{theorem}
The entropy can be computed as
\begin{equation}
\boxed{H(\mathcal{S} ) = -\sum^n_{i=1} \sum^{n}_{j=1} \pi^\star_i \: p_{ij} \log_2
p_{ij},}
\label{eq:economicentropy}
\end{equation}
where the elements $\{ \pi^\star_i \}$ are given by the eigenvector $1\mbs{\pi}^\star = \mbf{P}
\mbs{\pi}^\star$ with eigenvalue that equals 1. That is, it is a
nontrivial solution to $(\mbf{I} - \mbf{P})\mbs{\pi}^\star =
\mbf{0}$. The entropy is bounded by
$$0 \leq H(\mathcal{S} ) \leq \log_2 n,$$
and attains its maximum value when the transitions to each sector are equiprobable.
\end{theorem}

\begin{proof}
See \cite[thm.~4.2.4]{Cover&Thomas2012_elements}.
\end{proof}

The result \eqref{eq:economicentropy} provides an operational measure
that can readily be applied to existing input-output data collected by
national statistics services. Note that the $H(\mathcal{S} )$ is
invariant to the selection of the initial sector $s_0$.

\begin{remark}
\emph{Remark:} The economic system $\mathcal{S}$ can also be extended to include the
reproduction of its workforce. This is achieved by including the
households as a sector that outputs units of labour, with inputs in the
form of the real wage vector of the productive workforce, cf. \cite{Cockshott&Zachariah2006_hunting}.
\end{remark}

\section{Numerical examples}

In this section, we apply the $H(\mathcal{S})$ in
\eqref{eq:economicentropy} to input-output data from seven national
economies found in \cite{OECD_iodata,Korea_iodata}. The entropy measure
quantifies the complexity of an economy, and as it evolves over
time, $H(\mathcal{S})$ also registers the degree of structural economic change.

One limitation of the data is that the sectoral classification of the economy is not
consistent across all datasets. However, within each national economy
the classifications are sufficiently similar to provide meaningful
comparisons across time. Moreover, the number of basic sectors $n$ also
vary depending on classification level. To mitigate these effects we
also consider $H(\mathcal{S})$ as a percentage of its maximum
value. The results are presented in Table~\ref{tab:results} and Figure~\ref{fig:results}.
\begin{figure*}
\centering
   \begin{subfigure}[b]{0.48\textwidth}
   \includegraphics[width=1\linewidth]{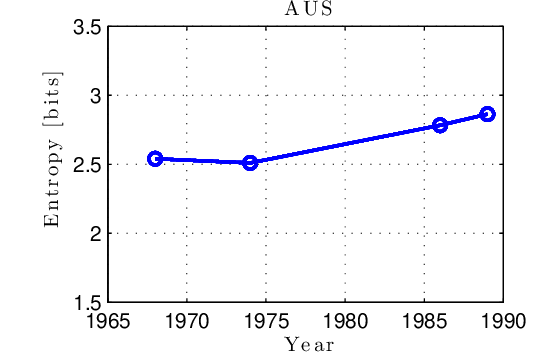}
\end{subfigure}
$\:$
\begin{subfigure}[b]{0.48\textwidth}
   \includegraphics[width=1\linewidth]{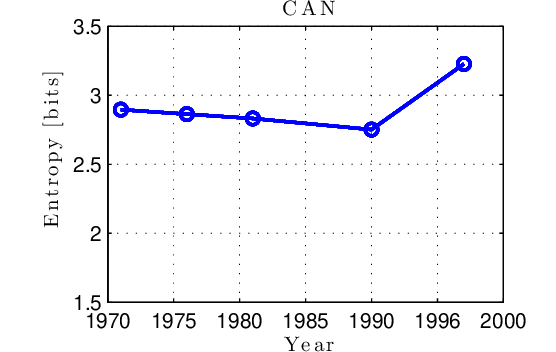}
\end{subfigure}

   \begin{subfigure}[b]{0.48\textwidth}
   \includegraphics[width=1\linewidth]{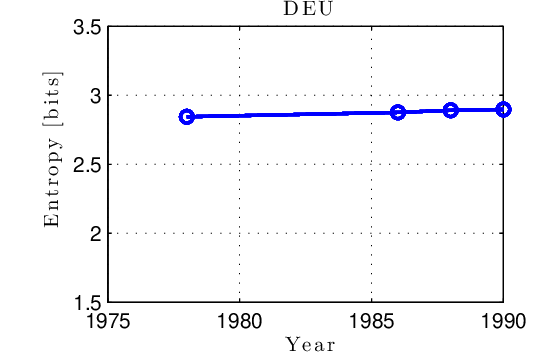}
\end{subfigure}
$\:$
\begin{subfigure}[b]{0.48\textwidth}
   \includegraphics[width=1\linewidth]{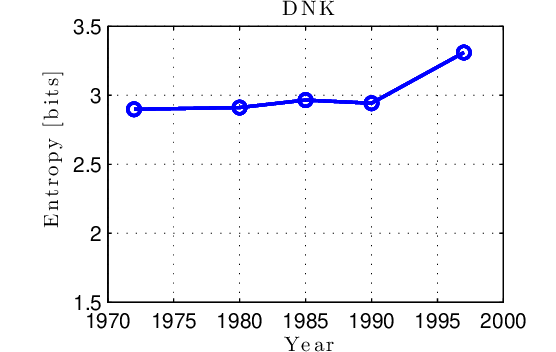}
\end{subfigure}

   \begin{subfigure}[b]{0.48\textwidth}
   \includegraphics[width=1\linewidth]{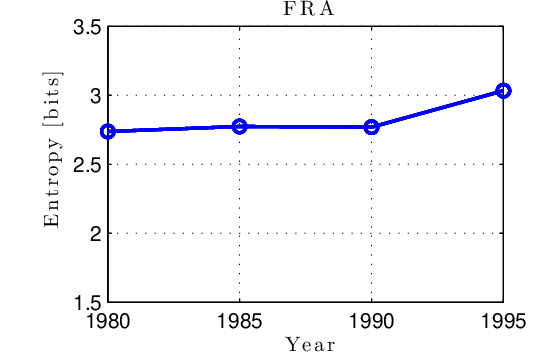}
\end{subfigure}
$\:$
\begin{subfigure}[b]{0.48\textwidth}
   \includegraphics[width=1\linewidth]{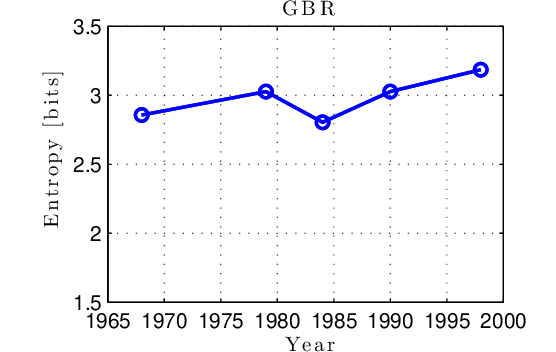}
\end{subfigure}

   \begin{subfigure}[b]{0.48\textwidth}
   \includegraphics[width=1\linewidth]{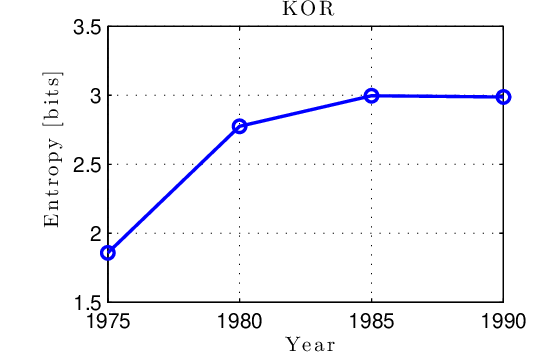}
\end{subfigure}
\caption[TEST]{Economic complexity $H(\mathcal{S})$ in bits}
\label{fig:results}
\end{figure*}

\begin{table}
\centering
\begin{subtable}{.5\textwidth}
\centering

\begin{tabular}{cccc}\hline
         & $n$ & [bits]   & [\%] \\ \hline
         
AUS 1968 &  30 & 2.540  &  51.76  \\
AUS 1974 &  30 & 2.509  &  51.14  \\
AUS 1986 &  30 & 2.782  &  56.70  \\
AUS 1989 &  30 & 2.863  &  58.34  \\ \hline

CAN 1971 & 31  & 2.896  &  58.46 \\
CAN 1976 & 31  & 2.863  &  57.78 \\
CAN 1981 & 31  & 2.831  &  57.15 \\
CAN 1990 & 31  & 2.751  &  55.52 \\
CAN 1997 & 32  & 3.226  &  64.53 \\ \hline

DEU 1978 & 29  & 2.844  &  58.53 \\
DEU 1986 & 29  & 2.876  &  59.20 \\
DEU 1988 & 29  & 2.890  &  59.48 \\
DEU 1990 & 29  & 2.897  &  59.63 \\ \hline

DNK 1972 & 28  & 2.898  &  60.29 \\
DNK 1980 & 28  & 2.911  &  60.56 \\
DNK 1985 & 28  & 2.965  &  61.67 \\
DNK 1990 & 28  & 2.942  &  61.20 \\
DNK 1997 & 36  & 3.310  &  64.02 \\
\end{tabular}

\end{subtable}
\begin{subtable}{.5\textwidth}
\centering

\begin{tabular}{cccc}\hline
         & $n$ & [bits]   & [\%] \\ \hline 

FRA 1980 & 31  & 2.737  &  55.25 \\
FRA 1985 & 31  & 2.773  &  55.98 \\
FRA 1990 & 31  & 2.768  &  55.88 \\
FRA 1995 & 37  & 3.033  &  58.23 \\ \hline

GBR 1968 & 31  & 2.857  &  57.67 \\
GBR 1979 & 31  & 3.027  &  61.10 \\
GBR 1984 & 31  & 2.804  &  56.59 \\
GBR 1990 & 31  & 3.027  &  61.10 \\
GBR 1998 & 37  & 3.185  &  61.14 \\ \hline

KOR 1975 & 346 & 1.858  &  22.03 \\
KOR 1980 & 359 & 2.775  &  32.69 \\
KOR 1985 & 363 & 2.996  &  35.23 \\
KOR 1990 & 369 & 2.988  &  35.04 \\ 
\end{tabular}

\end{subtable}
\caption{Economic complexity $H(\mathcal{S})$, in bits and
as a percentage of the maximum level $\log_2(n)$.}
\label{tab:results}
\end{table}

First, we observe that the complexity of each economy is low relative to
its maximum possible value. That is to say, relative to the complexity
that would prevail were all sectors equally likely as destinations of
the cross-sectoral outputs. For small input-output tables, the
complexity is about half of the maximum level. Across all of the datasets,
the complexity is sufficiently low so as to discriminate between at most
10 effective output destinations (3.3 bits) at most, and at least 4
destinations (1.9 bits).

Second, note that while the South Korean dataset (KOR) specifies the economy
at a much finer resolution than the other sets, the complexity is
still comparable to them. This is consistent with the results
presented in \cite{reifferscheidt2014average}, using disagregated
tables from the US economy. In that analysis it was found that the
number of intersectoral links grew at a rate below $\log(n)$. If we
express that structure as a directed graph with the nonzero elements
having equal weights then we would expect the $H(\mathcal{S})$
to be of order $\log_2(\log_2(n))$. Based on this result, an advanced economy
$\mathcal{S}$ with $n= 369$ sectors would be expected to have an entropy less
than 3.1 bits, which is what we observe here.

Third, $H(\mathcal{S})$ also registers structural economic
change. This is most dramatic in the South Korean economy. Between
1975 and 1980, the entropy increased by nearly one bit which
corresponds to a doubling of the complexity in five years. This is
reflective of its rapid state-led industrialization
process. By contrast, while economic complexity of the United Kingdom
(GBR) increases between 1968 and 1979, it exhibits a notable drop
by the mid-1980s, corresponding to a fifteen percent reduction of
complexity. This would reflect the significant decline of the
manufacturing industry, while the restructuring of the UK economy is
followed by a subsequent rise in complexity. The German economy (DEU)
shows a stable level of complexity in the period prior to
unification. The trend in the Canadian economy (CAN) was similar, until
the mid-1990s.

\section{Conclusions}

Starting from an input-output model of an economic system, we have
developed a measure of its complexity based on Shannon
entropy. Following Leontief's perspective of the production process as
a circular flow, we formuled the production process as a Markov
chain. This enabled us to derive an operational measure of economic
complexity as the average code length required to encode the
destination of an output in the production process.

We applied the measure to real data from seven national economies. It
was observed that the complexity of each economy is substantially below the
maximum possible value. The results for the larger data sets were also
consistent with scaling laws observed in other economies. In addition,
the measure was found to register structural economic changes of
industrialization and deindustrialization.

These results suggest that $H(\mathcal{S})$ is reasonably similar among
advanced industrialized economies (approximately 3 bits in the 1990s),
and therefore related to productivity levels and other economic
development indicators. While there is no clear relation between small variations
in the entropy and growth rates of outputs, the rapid structural
changes registered in $H(\mathcal{S})$ do appear to be associated
with changes in output trajectory of the system, as in the case of
South Korea and the United Kingdom mentioned above. Further research
using updated input-output data from the original national statistics
bureaus would be needed to assess this as well as the possible
noise fluctations in the estimated entropies. Such data would need to
extend the time-series into the last two decades. It would then be
possible to quantify the structural effects of the financial crisis of
2007-2008 and its repercussions.


\small{
\bibliographystyle{ieeetr}
\bibliography{refs_entropy,paul_bibliography}
}
\end{document}